\documentclass[11pt,letterpaper]{article}

\usepackage[utf8]{inputenc}
\usepackage[english]{babel}
\usepackage{authblk}

\usepackage[dvipsnames,usenames]{xcolor}
\usepackage[colorlinks=true,pdfpagemode=UseNone,urlcolor=RoyalBlue,linkcolor=RoyalBlue,citecolor=OliveGreen,pdfstartview=FitH]{hyperref}

\usepackage{subcaption}

\usepackage[margin=1in]{geometry}
\usepackage{amsfonts,amsmath,amsthm,amssymb,thmtools}
\usepackage{bm}
\usepackage{nicefrac}
\usepackage[ruled, vlined]{algorithm2e}

\declaretheorem{theorem}
\declaretheorem[sibling=theorem]{lemma}

\declaretheorem[sibling=theorem]{observation}
\declaretheorem[sibling=theorem]{proposition}

\declaretheorem[style=remark]{remark}

\usepackage{color}
\usepackage{todonotes}
\usepackage{soul}

\usepackage{booktabs}
\usepackage{caption}

\usepackage[numbers, sort]{natbib}
\usepackage{tcolorbox}
\usepackage{enumitem}
\usepackage{tikz}

\newcommand{\abs}[1]{\left|{#1}\right|}

\newcommand{\set}[1]{\left\{{#1}\right\}}

\newcommand{\dis}{\mathsf{dist}}
\newcommand{\sep}{\mathsf{sep}}
\newcommand{\Ball}{\mathsf{Ball}}

\newcommand{\bundle}{\mathsf{b}}
\newcommand{\Bundle}{\mathsf{Bundle}}

\title{
    A Randomized Algorithm for Single-Source Shortest Path on Undirected Real-Weighted Graphs
}

\author[1]{Ran Duan \thanks{duanran@mail.tsinghua.edu.cn}}
\author[1]{Jiayi Mao \thanks{mjy22@mails.tsinghua.edu.cn}}
\author[2]{Xinkai Shu \thanks{xkshu@cs.hku.hk}}
\author[1]{Longhui Yin \thanks{ylh21@mails.tsinghua.edu.cn}}

\affil[1]{Institute for Interdisciplinary Information Sciences, Tsinghua University}
\affil[2]{The University of Hong Kong}

\begin{document}

\begin{titlepage}
    \thispagestyle{empty}
    \maketitle
    \begin{abstract}
        \thispagestyle{empty}
        In undirected graphs with real non-negative weights, we give a new randomized algorithm for the single-source shortest path (SSSP) problem with running time $O(m\sqrt{\log n \cdot \log\log n})$ in the \emph{comparison-addition} model. 
This is the first algorithm to break the $O(m+n\log n)$ time bound for real-weighted sparse graphs by Dijkstra's algorithm with Fibonacci heaps. 
Previous undirected non-negative SSSP algorithms give time bound of $O(m\alpha(m,n)+\min\{n\log n, n\log\log r\})$ in comparison-addition model, where $\alpha$ is the inverse-Ackermann function and $r$ is the ratio of the maximum-to-minimum edge weight [Pettie \& Ramachandran 2005], and linear time for integer edge weights in RAM model [Thorup 1999]. Note that there is a proposed complexity lower bound of $\Omega(m+\min\{n\log n, n\log\log r\})$ for hierarchy-based algorithms for undirected real-weighted SSSP [Pettie \& Ramachandran 2005], but our algorithm does not obey the properties required for that lower bound. As a non-hierarchy-based approach, our algorithm shows great advantage with much simpler structure, and is much easier to implement.

    \end{abstract}
\end{titlepage}

\section{Introduction}
\label{sec:introduction}

Shortest path is one of the most fundamental problems in graph theory, and its algorithms lie at the core of graph algorithm research. In a graph $G=(V,E,w)$ with $m=|E|$, $n=|V|$ and non-negative edge weight $w:E\rightarrow \mathbb{R}_{\geq 0}$, single-source shortest path (SSSP) problem asks for the distances from a given source $s \in V$ to all other vertices. 
Dijkstra's algorithm \cite{Dij59} computes the distances $\dis(s,u)$ by dynamic programming. For each vertex $u$, it maintains a temporal distance $d(u)$, which represents the shortest path from $s$ to $u$ only passing through the vertices in current $S$, where $S$ is the set of \emph{visited} vertices during algorithm process. In each round of iteration it selects vertex $u$ with the smallest $d(u)$ from the \emph{unvisited} nodes. Finally when $S=V$, $d(u)=\dis(s,u)$ for all vertex $u$.
Advanced data structures with amortized $O(1)$ time for insertion and decrease-key, and $O(\log n)$ for extract-min, called Fibonacci heap \cite{FT87} and relaxed heap \cite{DJG88}, make the time bound for Dijkstra's algorithm to $O(m+n\log n)$.
This time bound is in the \emph{comparison-addition} model where only comparison and addition operations on edge weights are allowed and considered as unit-time operations, which is the most common model for real number inputs. 
For undirected graphs, \citet*{PR05} proposed an SSSP algorithm with running time $O(m\alpha(m,n)+\min\{n\log n, n\log\log r\})$ in the comparison-addition model, where $\alpha$ is the inverse-Ackermann function and $r$ bounds the ratio of any two edge weights. However, no SSSP algorithm faster than $O(m+n\log n)$ has been found for real-weighted graphs without ratio constraints, both for undirected and directed graphs.

A byproduct of Dijkstra's algorithm is the sorting of all vertices by their distances from $s$, but the lower bound of $\Omega(n\log n)$ lies for comparison-based sorting algorithms. Researchers used to believe that this sorting bottleneck existed for many graph problems,  
and breaking this bottleneck is an important and interesting direction. Yao \cite{Yao75} gave a minimum spanning tree (MST) algorithm with running time $O(m\log\log n)$, citing an unpublished result of $O(m\sqrt{\log n})$ by Tarjan. The current best results for MST are the randomized linear time algorithm~\cite{KKT95}, the deterministic $O(m\alpha(m,n))$-time algorithm~\cite{Cha00}, and a deterministic algorithm with proven optimal (but unknown) complexity~\cite{PR02}. In the bottleneck path problem, we want to find the path maximizing the minimum edge weight on it between two vertices. \citet*{GT88} gave an $O(m\log^* n)$-time algorithm for s-t bottleneck path problem in directed graphs, which was later improved to randomized $O(m\beta(m,n))$ time~\cite{CKTZZ}. For single-source all-destination bottleneck path problem in directed graphs, there is a recent result of $O(m\sqrt{\log n})$-time randomized algorithm by \citet{DLWX}. For single-source nondecreasing path problem, Virginia V.Williams~\cite{VW10} proposed an algorithm with time bound $O(m\log\log n)$. All the results above are comparison-based though, the techniques in these works, such as local construction or divide-and-conquer approach, hardly works for the shortest path problem. Therefore it remains how to break the sorting bottleneck for SSSP.

\subsection{Our Results}
In this paper we propose the first SSSP algorithm for undirected real-weighted graphs that breaks the sorting bottleneck.

\begin{theorem}
\label{thm:breaking-sorting-bottleneck}
In an undirected graph $G=(V,E,w)$ with nonnegative edge weights $w:E\rightarrow\mathbb{R}_{\geq 0}$, there is a comparison-addition based Las-Vegas randomized algorithm that solves the single-source shortest path problem in $O(m\sqrt{\log n\cdot\log\log n})$ time, in which the results are always correct and it can achieve this time bound with high probability. The time complexity can be improved to $O(\sqrt{mn\log n}+n\sqrt{\log n\log\log n})$ when $m=\omega(n)$ and $m=o(n\log n)$.
\end{theorem}

Note that there is a (worst-case) lower bound of $\Omega(m+\min\{n\log n, n\log\log r\})$ in~\cite{PR05} for ``hierarchy-based'' algorithms for undirected real-weighted SSSP, but our algorithm is randomized and not hierarchy-based. See Remark~\ref{rmk:lower-bound} for discussions.

\paragraph{Technical Overview.}

The bottleneck of Dijkstra-based algorithm is the priority queue. For this reason, we only add a fraction of vertices into the priority queue. As in many works on distance oracles or spanners, we sample a subset of vertices $R$, and the heap is only for vertices in $R$, then we ``bundle'' every other vertex $v$ to its nearest vertex in $R$, which is called $\bundle(v)$. Then define $\Ball(v)$ to be the set of vertices closer than $\bundle(v)$ to $v$. Since the algorithm doesn't know the correct order of most vertices on a shortest path, relaxing only neighbours as in Dijkstra's algorithm doesn't work. So when popping a vertex $u\in R$ from the heap, we also deal with vertices $v$ which are bundled to $u$. In an undirected graph, this also implies that $|\dis(s,u) - \dis(s,v)|$ is not large. Here we relax $v$ from vertices in $\Ball(v)$ and their neighbors, then from $v$ we relax neighbors of $v$ and vertices in their balls. (To make it easier to describe, we first change the graph to a constant-degree graph with $O(m)$ vertices.) Details of the algorithm will be discussed in Section \ref{sec:algorithm}, as well as the analysis of correctness and running time. The detailed construction of bundles so that the algorithm can achieve the time bound w.h.p. will be introduced in Section \ref{sec:improved-construction-of-bundles}. The improvement of time complexity by relaxing the constant-degree constraint will be discussed in Section~\ref{sec:disussion}.

\subsection{Other Related Works}

Existence of algorithms better than $O(m+n\log n)$ for real-weighted SSSP has been open for long. Pettie and Ramachandran's algorithm~\cite{PR05} works better than $O(n\log n)$ if the ratio between maximum and minimum edge weights is not very large. For the integer-weighted case, random access machine (RAM) model is usually adopted, where multiplications, shifts and Boolean operations on edge weights are allowed. There are many works on improving heaps and SSSP algorithms on RAM model with integer weights~\cite{FW93,FW94, Thorup96, Raman96, Raman97, TM00, HT20}.
Finally Thorup gave a linear-time algorithm for undirected graphs~\cite{Thorup00} and $O(m+n\log\log\min\{n,C\})$ for directed graphs~\cite{Thorup04} where $C$ is the maximum edge weight. 
Recently, almost linear time $O(m^{1+o(1)}\log C)$ algorithms for SSSP with negative weights are also discovered~\cite{flow,BNW22}.

\emph{All-pair shortest path (APSP)} problem requires the shortest path between every pair of vertices $u,v$ in the graph $G$. We can run Dijkstra's algorithm~\cite{Dij59} from all vertices which will have running time $O(mn+n^2\log n)$, or use Floyd-Warshall algorithm~\cite{RF62,WS62} with running time $O(n^3)$. Researchers have made many improvements since then~\cite{FM76,WD90,ZU04,CT08,HT12}, but there is still no truly subcubic time ($O(n^{3-\epsilon})$ for some constant $\epsilon>0$) APSP algorithm for real-weighted graphs or even graphs with integer weights in $[0,n]$. Williams~\cite{Williams14} gave an APSP algorithm with running time $n^3/2^{\Theta(\sqrt{\log n})}$ for real-weighted graphs. For undirected real-weighted graphs, Pettie and Ramachandran's APSP algorithm~\cite{PR05} runs in $O(mn\log\alpha(m,n))$ time, and for directed real-weighted graphs, Pettie~\cite{Pettie04} gave an APSP algorithm in $O(mn+n^2\log\log n)$ time.

\section{Preliminaries}
\label{sec:preliminaries}

In this paper we work on an undirected graph $G=(V, E, w)$ with vertex set $V$, edge set $E\subseteq V^2$ and non-negative weight function $w: E \to \mathbb{R}_{\geq 0}$, also denoted by $w_{uv}$. In an undirected graph $w_{uv} = w_{vu}$ holds for all edges $(u, v) \in E$. We denote $n = \abs{V}$, $m=\abs{E}$ as the number of vertices and edges in the graph, and $N(u) = \{v: (u, v) \in E\}$ as the set of neighbors of $u$. For two vertices $u, v\in V$, $\dis_{G}(u,v)$ is length of the shortest path connecting $u$ and $v$, namely the distance of $u$ and $v$ in graph $G$. The subscript $G$ is omitted when the context is clear. Let $s$ be the source vertex. The target of our algorithm is to find $\dis(s, v)$ for every $v\in V$. Without loss of generality we assume that $G$ is connected, so $m\geq n-1$.

\paragraph{Constant-Degree Graph.} Throughout the paper we need a graph with constant degree. To accomplish this, given a graph $G$, we construct $G'$ by a classical transformation (see~\cite{Frederickson83}):
\begin{itemize}
    \item Substitute each vertex $v$ with a cycle of $\abs{N(v)}$ vertices $x_{vw}$ ($w\in N(v)$) connected with undirected zero-weight edges, that is, for every neighbor $w$ of $v$, there is a vertex $x_{vw}$ on this cycle.
    \item For every edge $(u,v)$ in $G$, add an undirected edge between corresponding vertices $x_{uv}$ and $x_{vu}$ with weight $w_{uv}$.
\end{itemize}
We can see distance $\dis_{G'}(x_{uu'}, x_{vv'}) = \dis_G(u, v)$ for arbitrary $u'\in N(u)$ and $v'\in N(v)$. Each vertex in $G'$ has degree at most 3, while $G'$ being a graph with $O(m)$ vertices and $O(m)$ edges.

\paragraph{Comparison-Addition Model.} In this paper our algorithm works under \emph{comparison-addition} model, 
in which real numbers are subject to only comparison and addition operations. In this model, each addition and comparison takes unit time, and no other computations on edge weights are allowed. 

\paragraph{Fibonacci Heap.} Under such a model, it is possible to construct a Fibonacci heap $H$ that spends amortized $O(1)$ time for initialization, insertion, decrease-key operations, and $O(\log \abs{H})$ time for each extract-min operation~\cite{FT87}. When we extract the minimum element from the heap, we also call that element ``popped'' from the heap.




\section{Main Algorithm}
\label{sec:algorithm}

In the following sections we assume that $G$ is connected, and each vertex in $G$ has degree no more than $3$. (There are $O(m)$ vertices and $O(m)$ edges in $G$, but we still use $O(\log n)$ which is equivalent to $O(\log m)$ where $n$ is the number of vertices in the original graph without degree constraints.)


Our algorithm is based on original Dijkstra's algorithm \cite{Dij59}. Since the main bottleneck of Dijkstra's algorithm is the $O(\log n)$ time per every vertex extracted from the Fibonacci heap~\cite{FT87}, we only insert a subset $R \subseteq V$ of vertices into the heap. Each remaining vertex $v\in V \setminus R$ is \emph{bundled} to its closest vertex in $R$. Throughout the algorithm, vertices are updated only when some vertex $u\in R$ is popped from the heap. Our algorithm consists of two stages: \emph{bundle construction} and \emph{Bundle Dijkstra}, whose details will be introduced in Section \ref{sec:constructing-bundles} and \ref{sec:bundle-dijkstra} respectively.

To demonstrate the main idea of our algorithm, in this section we first give an algorithm that runs in expected $O(m\sqrt{\log n\cdot \log\log n})$ time but not ``with high probability''. In Section \ref{sec:improved-construction-of-bundles} we give an improved version of the bundle construction stage, leading to an algorithm that runs in $O(m\sqrt{\log n\cdot \log\log n})$ time with high probability. Both algorithms always give correct answers.

\subsection{Bundle Construction}
\label{sec:constructing-bundles}

A simple version of bundle construction works as follows\footnote{One may notice that sampled set, closest sampled vertex and balls are common techniques in papers on shortest path algorithms, distance oracles and spanners, and there are deterministic construction algorithms for such ``dominating set'' (e.g. \cite{ACM96}), but the extra $O(\log n)$ factor for deterministic approaches introduced on the size of dominating set or construction time is not affordable here.} ($k$ is a parameter to be determined later):
\begin{itemize}
    \item Independently sample each vertex $v\in V\setminus \{s\}$ with probability $\frac{1}{k}$ to form set $R$, then add $s$ into $R$.
    \item For each vertex $v \notin R$, run Dijkstra's algorithm started from $v$ until first vertex of $R$ is extracted from the heap, denoted by $\bundle(v)$. Therefore $\bundle(v)$ is one of the closest vertices in $R$ to $v$, i.e., $\bundle(v) \in \arg\min_{u\in R}\dis(u, v)$. We say that $v$ is \emph{bundled} to $\bundle(v)$.
    \item For each $u\in R$, let $\bundle(u) = u$, and $\Bundle(u) = \{v: u = \bundle(v)\}$ be the set of vertices bundled to $u$. By definition, $\{ \Bundle(u)\}_{u\in R}$ forms a partition of the vertex set $V$.
    \item For each vertex $v \notin R$, define $\Ball(v) = \{w\in V: \dis(v, w) < \dis(v, \bundle(v))\}$, that is, the set of vertices closer to $v$ than its bundled vertex $\bundle(v)$. In the previous Dijkstra's algorithm we can get $\Ball(v)$ and also values of $\dis(v, w)$ for all $w\in \Ball(v)\cup\{\bundle(v)\}$.
\end{itemize}

\paragraph{Time Analysis of Bundle Construction.} For each vertex $v\notin R$, without loss of generality we assume its Dijkstra's algorithm breaks tie in a deterministic way. Therefore, the order of vertices extracted from heap is fixed.

We can see $\mathbb{E}[|R|]=O(m/k)$. For each vertex $v\notin R$, let $S_v$ be the set of vertices extracted before its Dijkstra's algorithm stops, then $\Ball(v) \subsetneq S_v$. By definition of $R$, $\abs{S_v}$ follows geometric distribution with success probability $\frac{1}{k}$, thus $\mathbb{E}[\abs{S_v}]=k$ and $\mathbb{E}[\abs{\Ball(v)}] \leq k$.
By constant degree property, the number of vertices ever added into the heap is also $O(\abs{S_v})$, so the total time of the bundle construction is $O(\sum_{v\in V\setminus R} \mathbb{E}[\abs{S_v}\log \abs{S_v}])=O(mk\log k)$ in expectation.

\begin{remark}
One may argue that $x\log x$ is a convex function so that $\mathbb{E}[\abs{S_v}\log \abs{S_v}] = O(k\log k)$ does not trivially hold due to Jensen's inequality. Thanks to an anonymous reviewer, just notice that $\mathbb{E}[\abs{S_v}^2] = 2k^2 - k = O(k^2)$ and $\sqrt{x}\log{x}$ is a concave function when $x\geq 1$.

\end{remark}


\subsection{Bundle Dijkstra}
\label{sec:bundle-dijkstra}
Given the set $R$ and the partition of bundles, the main algorithm works as follows, with pseudocode given in Algorithm~\ref{alg:bundle-Dijkstra}:

Initially we set $d(s)=0$ and $d(v)=+\infty$ for all other vertex $v$, and insert all vertices of $R$ into a Fibonacci heap~\cite{FT87}. Whenever we pop a vertex $u\in R$ from the heap, we update the distances by the following steps. (Here relaxing a vertex $v$ by a value $D$ means that we update $d(v)$ by $\min\{d(v), D\}$.)
\begin{enumerate}
    \item For every vertex $v$ bundled to $u$, we need to find the exact value of $\dis(s,v)$. First relax $v$ by $d(u)+\dis(u,v)$; then for every vertex $y\in\Ball(v)$, relax $v$ by $d(y)+\dis(y,v)$; and for every $z_2 \in \Ball(v)\cup\{v\}$ and $z_1\in N(z_2)$, relax $v$ by $d(z_1)+w_{z_1,z_2}+\dis(z_2,v)$. That is, we update $d(v)$ by its bundled vertex $u$, vertices in $\Ball(v)$, and vertices neighboring to $v$ and $\Ball(v)$.
    \item After updating $d(x)$ for every $x\in\Bundle(u)$, we update the vertices $y\in N(x)$ and vertices $z_1\in\Ball(y)$. That is, relaxing $y$ by $d(x)+w_{x,y}$ for all $y\in N(x)$ and then relaxing $z_1$ by $d(x)+w_{x,y}+\dis(y,z_1)$ for all $z_1\in\Ball(y)$.
    \item Whenever we update a vertex $v\notin R$, we also relax its bundled vertex $\bundle(v)$ by $d(v)+\dis(v,\bundle(v))$. (But later we will see this is only needed in Step 2 but not Step 1, since in Step 1 $v$ is bundled to $u$, but the distance $\dis(s,u)$ is already found when popping $u$ from the heap.)
\end{enumerate}

The following observation holds naturally from the algorithm.

\begin{observation}\label{obs:d-upper-bound}
    $d(v) \geq \dis(s, v)$ always holds for all $v \in V$.
\end{observation}

\begin{algorithm}
    \caption{\textsc{BundleDijkstra}$(G,s)$}
    \label{alg:bundle-Dijkstra}
    \SetKwInOut{Input}{Input}\SetKwInOut{Output}{Output}
    \SetKwProg{Func}{function}{:}{}
    \Input{A graph $G=(V,E,w)$ and starting vertex $s\in V$}
    \Output{Distance $d(v)$ from $s$ to $v$ for every vertex $v\in V$}
    
    \nl Construct Bundles as described in Section~\ref{sec:constructing-bundles}\;
    
    \nl Set label $d(s) \gets 0$ and $d(v) \gets +\infty$ for all $v \in V\setminus \{s\}$\;
    
    \nl Initialize Fibonacci heap $H$ with all vertices of $R$ and key $d(\cdot)$\;
    
    \nl \While{$H$ is not empty}{
 
    \nl 	$u \gets H$.\textsc{ExtractMin}$()$\; \label{line:bundle-Dijkstra-extractmin}
        
    \nl     \ForEach(\tcp*[f]{Step 1}){$v \in \Bundle(u)$} {
        
    \nl         \textsc{Relax}($v$, $d(u) + \dis(u, v)$)\;\label{line:bundle-Dijkstra:update-uv}  
            
    \nl         \ForEach{$y \in \Ball(v)$}{
    
    \nl             \textsc{Relax}($v$, $d(y) + \dis(y, v)$)\; \label{line:bundle-Dijkstra-update-yv}
            }

    \nl         \ForEach{$z_2 \in \Ball(v)\cup\{v\}$}{
    
    \nl             \ForEach{$z_1 \in N(z_2)$}{
    
    \nl                 \textsc{Relax}($v$, $d(z_1) + w_{z_1, z_2} + \dis(z_2, v)$)\; \label{line:bundle-Dijkstra:update-z1v}
                }
            }
        }
        
    \nl     \ForEach(\tcp*[f]{Step 2}){$x \in \Bundle(u)$}{
            
    \nl         \ForEach{$y \in N(x)$}{

    \nl             \textsc{Relax}($y$, $d(x) + w_{x, y}$)\;\label{line:bundle-Dijkstra:update-xy}
                
    \nl             \ForEach{$z_1 \in \Ball(y)$}{
                
    \nl                 \textsc{Relax}($z_1$, $d(x) + w_{x, y} + \dis(y, z_1)$)\; \label{line:bundle-Dijkstra:update-xz1}

                        
                        
                }
            }
        }
    }

    \Func{\textsc{\textrm{Relax($v$, $D$)}}}{
         \If{$D < d(v)$}{
            $d(v)\gets D$\;
            \If{$v\in H$}{
                $H$.\textsc{DecreaseKey}($v$, $D$)}
            \ElseIf{$v\notin R$}{
                \textsc{Relax}($\bundle(v), d(v) + \dis(v, \bundle(v))$)\tcp*[f]{Step 3}
            }
        }
        
    }
\end{algorithm}

\paragraph{Time Analysis for Bundle Dijkstra.} For the Bundle Dijkstra stage, only vertices in $R$ are inserted into heap, thus the extract-min operation only takes $O(\abs{R}\log n)$ time in total. Since every vertex in $V\setminus R$ only appears once as $v$ and $x$ in Step 1 and Step 2, respectively, and by constant degree property, every vertex appears constant times as the vertex $y\in N(x)$ in Step 2. We can see the number of vertices $z_1,z_2$ in Step 1 for every $v$ is $O(|\Ball(v)|)$, and the number of vertices $z_1$ in Step 2 for every $y$ is $O(|\Ball(y)|)$. Also note that the recursive call of \textsc{Relax} in Step 3 can only recurse once since $\bundle(v)\in R$. So the total time for Step 1, 2 and 3 is $O(\sum_{v\in V\setminus R}|\Ball(v)|)$. Thus, the time of the bundle Dijkstra stage is $\mathbb{E}[O(|R|\cdot\log n+\sum_{v\in V\setminus R}\abs{\Ball(v)})] = O(\frac{m}{k}\log n + mk)$ in expectation.

Now, we can see that the expected total time of the two stages is $O(\frac{m}{k}\log n + mk\log k)$, which is minimized to $O(m\sqrt{\log n\cdot \log \log n})$ if we choose $k = \sqrt{\frac{\log n}{\log \log n}}$. We move to explain our main ideas of the correctness proof. A formal proof is given in Section~\ref{sec:correctness-analysis}.

\paragraph{Main ideas.} The following propositions hold in the algorithm. (Here the iteration of $u$ means the iteration performed when popping $u\in R$; a real distance $\dis(s,v)$ is found means $d(v)=\dis(s,v)$ already holds.)

\begin{proposition}\label{inv:R-vertex}
   When popping $u\in R$ from the heap, its distance $\dis(s,u)$ is already found.
\end{proposition}
\begin{proposition}\label{inv:other}
   After Step 1 in the iteration of $u$, $\dis(s,v)$ for all $v\in \Bundle(u)$ are found.
\end{proposition}

The following lemmas contain the main ideas of the algorithm.
\begin{lemma}\label{lemma:R-vertex}
    For any vertex $u\in R$ and any path $P$ from $s$ to $u$, if $P$ goes through vertex $y$, $\dis(s,\bundle(y))$ is at most the length of $P$.
\end{lemma}
\begin{proof}
    $\dis(s, y)$ is at most the length of subpath of $P$ from $s$ to $y$. By definition of $\bundle(y)$, $\dis(y, \bundle(y))$ is at most the length of subpath of $P$ from $y$ to $u$. Concatenating two subpaths together, $\dis(s, \bundle(y)) \leq \dis(s, y) + \dis(y, \bundle(y))$ is at most the length of $P$.
\end{proof}

Lemma~\ref{lemma:R-vertex} shows that for any vertex $u\in R$, the shortest path from $s$ to $u$ only contains vertices $y$ with $\dis(s, \bundle(y)) \leq \dis(s, u)$. This is the intuition why vertices of $R$ are popped in increasing order of $\dis(s,\cdot)$. However, the shortest path from $s$ to some vertex $v\in \Bundle(u)$ may go through some vertex $y$ with $\dis(s, \bundle(y)) \geq \dis(s, u)$, that is, $\bundle(y)$ is still not popped from the heap. But surprisingly, with the ideas of Lemma~\ref{lemma:other} we can deal with this case even before the iteration of $\bundle(y)$.

\begin{lemma}\label{lemma:other}
    For a vertex $v\notin R$, if the shortest path from $s$ to $v$ is shorter than $\dis(s,\bundle(v))+\dis(\bundle(v),v)$, and it goes through a vertex $y$ (other than $v$) such that $\dis(s,\bundle(y))\geq \dis(s,\bundle(v))$, then on the shortest path from $y$ to $v$ there are two adjacent vertices $z_1,z_2$ such that $z_1\in\Ball(y)\cup\{y\}$ and $z_2\in\Ball(v)\cup\{v\}$.
\end{lemma}

\begin{proof}
    We have $\dis(y,v)=\dis(s,v)-\dis(s,y)$ and $\dis(s,v)<\dis(s,\bundle(v))+\dis(\bundle(v),v)$. By triangle inequality, $\dis(s,y)\geq \dis(s,\bundle(y))-\dis(y,\bundle(y))$, and by $\dis(s,\bundle(y))\geq \dis(s,\bundle(v))$,
    $$\dis(y,v)<\dis(\bundle(v),v)+\dis(y,\bundle(y))+\dis(s,\bundle(v))-\dis(s,\bundle(y))\leq \dis(\bundle(v),v)+\dis(y,\bundle(y))$$
    Let $z_1$ be the last vertex on the shortest path from $y$ to $v$ satisfying $\dis(y,z_1)<\dis(y,\bundle(y))$, so $z_1\in\Ball(y)$. Then $z_2$ will be the next vertex after $z_1$, so $\dis(y,z_2)\geq\dis(y,\bundle(y))$, and $\dis(z_2,v)<\dis(v,\bundle(v))$, so $z_2\in\Ball(v)$. (If $\dis(y,\bundle(y))=0$ then $z_1=y$, and if $\dis(y,v)<\dis(y,\bundle(y))$ then $z_2=v$ and $z_1$ is the vertex before $v$.)
\end{proof}

Then we can see Proposition \ref{inv:R-vertex} and \ref{inv:other} hold throughout the algorithm iteratively: (A formal proof will be given in Section~\ref{sec:correctness-analysis}.)
\begin{itemize}
    \item When we pop the source $s$ from the heap, $d(s)=0$, and the distances $\dis(s,v)$ for all $v\in \Bundle(s)$ are found in the bundle construction step, and can be put in $d(v)$ in Step 1. 
   \item Assume Proposition~\ref{inv:other} holds for the first $i$ vertices popped, so the real distances for all vertices bundled to popped vertices are found after Step 1. By Step 2 and 3, we can see for all unpopped $u\in R$, the distance $\dis(s,u)$ can be found if the shortest path from $s$ to $u$ does not go through vertices bundled to other unpopped vertices in the heap. If the next popped vertex $u'$ does not satisfy this, let $y$ be the first vertex on the shortest path from $s$ to $u'$ which is bundled to an unpopped vertex $\bundle(y)$ other than $u'$, so $\dis(s,\bundle(y))$ can be found. By Lemma~\ref{lemma:R-vertex}, $\dis(s,\bundle(y))\leq \dis(s,u')$, so if $d(u')>\dis(s,u')$, $\bundle(y)$ will be the next popped vertex. Thus, $\dis(s,u')$ for the next popped vertex $u'$ is found before it is popped.
   \item If an unpopped vertex $u'\in R$ is updated in the iteration of popped vertex $u$, the new path to $u'$ must go through a vertex in $\Bundle(u)$. By Lemma~\ref{lemma:R-vertex}, $d(u')$ cannot be updated to be smaller than $\dis(s,u)$, so the unpopped vertices must have longer or equal distances than any popped vertex.
   \item Thus when popping a vertex $u\in R$, its distance $\dis(s,u)$ is already found. For all vertex $v\in\Bundle(u)$, if $\dis(s,v)$ is not directly obtained by $d(u)+\dis(u,v)$, that is, $\dis(s,v)<\dis(s,u)+\dis(u,v)$, let $x$ be the last vertex on the shortest path from $s$ to $v$ such that $\bundle(x)$ is popped before $u$, and let $y$ be the next vertex after $x$. We can see $\dis(s,\bundle(y))\geq \dis(s,u)$, so by Lemma~\ref{lemma:other}, we get such $z_1$ and $z_2$. Then from Proposition~\ref{inv:other} $\dis(s,x)$ can be found in Step 1 in the iteration of $\bundle(x)$, then $\dis(s,z_1)$ can be found in Step 2 of that iteration. In this iteration of $u$, $\dis(s,v)$ can be set to $\dis(s,z_1)+w_{z_1,z_2}+\dis(z_2,v)$ in Step 1, so Proposition~\ref{inv:other} still holds after this iteration.
\end{itemize}

\subsection{Proof of Correctness}
\label{sec:correctness-analysis}

We give a formal proof based on the pseudocode of Algorithm~\ref{alg:bundle-Dijkstra}. Define $u_i\in R$ as the vertex extracted in the $i$-th iteration of while-loop in Algorithm \ref{alg:bundle-Dijkstra}. Our key lemma in the following shows the main properties of the algorithm, therefore Bundle Dijkstra is correct no matter how $R$ is chosen.

\begin{lemma}
    \label{lem:shortest-path-iteration}
     The following properties hold for any $i\geq 1$ in Bundle Dijkstra (Algorithm \ref{alg:bundle-Dijkstra}):
    \begin{enumerate}
        \item When $u_i$ is extracted from the heap, $d(u_i) = \dis(s, u_i)$ holds. \label{prop:du-correct-before} 
        
        \item After $i$-th iteration of the while-loop, $d(u) \geq d(u_i)$ for all $u \in R \backslash \set{u_j}_{j\leq i}$. \label{prop:heap-pop-increase}
        
        \item After Step 1 of $i$-th iteration of the while-loop, $d(v)=\dis(s,v)$ for all $v \in \Bundle(u_i)$. \label{prop:dv-correct}
    \end{enumerate}
    
\end{lemma}

\begin{proof}
    We shall prove the lemma by induction on $i$.
    
    The lemma holds for $i = 1$ since $d(s) = 0$ and $d(v) = \dis(s, v)$ for all $v \in \Bundle(s)$ after Line \ref{line:bundle-Dijkstra:update-uv}.
    
    Suppose the lemma holds for every $i \leq t-1$, consider the case $i=t$.
    
    \begin{enumerate}
        \item Consider a shortest path $P$ from $s$ to $u_t$. Let $x$ be the last vertex on $P$ such that $x\in \Bundle(u_j)$ for some $j < t$, and $y$ be the next one after $x$, hence $y \in \Bundle(u)$ for some $u \in R \backslash \set{u_\ell}_{\ell < t}$. By Property \ref{prop:dv-correct} of induction hypothesis $d(x) = \dis(s, x)$ after Step 1 of $j$-th iteration. After that the algorithm updates $d(y)$ in Line \ref{line:bundle-Dijkstra:update-xy} since $y \in N(x)$, and further $d(u)$.
        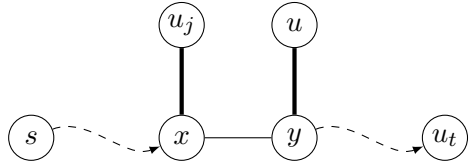
\begin{figure}[ht]
            \centering
            \begin{tikzpicture}
                \tikzstyle{mycircle}=[circle, minimum size = 6mm]
                \node [mycircle, draw, label=center:$s$] (s) at (-1.5, 0) {};
                \node [mycircle, draw, label=center:$x$] (x) at (0.5, 0) {};
                \node [mycircle, draw, label=center:$y$] (y) at (2, 0) {};
                \node [mycircle, draw, label=center:$u_j$] (uj) at (0.5, 1.5) {};
                \node [mycircle, draw, label=center:$u$] (u) at (2, 1.5) {};
                \node [mycircle, draw, label=center:$u_t$] (ut) at (4, 0) {};
                \draw [-latex, dashed] (s) .. controls (-0.75, 0.3) and (-0.25, -0.3) .. (x);
                \draw (x) -- (y);
                \draw [ultra thick] (x) -- (uj);
                \draw [ultra thick] (y) -- (u);
                \draw [-latex, dashed] (y) .. controls (2.75, 0.3) and (3.25, -0.3) .. (ut);
            \end{tikzpicture}
            \caption{Illustration of the shortest path from $s$ to $u_t$.}
        \end{figure}
        
        Therefore after $(t-1)$-th iteration $d(y) = \dis(s, x) + \dis(x, y) = \dis(s, y)$ and $d(u) \leq \dis(s, y) + \dis(y, u)$. Further:
        \begin{align*}
            d(u) \leq ~ & \dis(s, y) + \dis(y, u) \\
            \leq ~ &\dis(s, y) + \dis(y, u_t) \tag{$\bundle(y) = u$}\\
            = ~ &\dis(s, u_t) \tag{$y$ on shortest path}\\
            \leq ~ & d(u_t) \tag{Observation \ref{obs:d-upper-bound}}.
        \end{align*}

        On the other hand, the algorithm extracts $u_t$ from Fibonacci heap $H$ immediately after $(t-1)$-th iteration, thus $d(u_t) \leq d(u)$. So all the inequalities above should be equations, thus $d(u_t) = \dis(s, u_t)$.
        
        \item When executing Line \ref{line:bundle-Dijkstra-extractmin} of $t$-th iteration, $d(u) \geq d(u_t)$ holds for every $u \in R \backslash \set{u_j}_{j\leq t}$ since $H$ is a Fibonacci heap. 
        Suppose $d(u) < d(u_t)$ for some $u \in R \backslash \set{u_j}_{j\leq t}$ after $t$-th iteration. The further updates on $d(u)$ must start from $d(x)$ for some $x \in \Bundle(u_t)$. For last such update, applying Lemma \ref{lemma:R-vertex} on this path from $s$ to $x$ then to $u$, we have:
        \begin{align*}
            d(u) \geq ~ & \dis(s, u_t) \tag{Lemma \ref{lemma:R-vertex}}\\
            = ~ & d(u_t), \tag{Property \ref{prop:du-correct-before}}
        \end{align*}
        leading to contradiction.
        
        \item We want to show that $d(v) = \dis(s, v)$ holds for all $v\in \Bundle(u_t)$.
        Suppose there exists a vertex $v\in \Bundle(u_t)$ such that $d(v)>\dis(s,v)$ after Step 1. Denote $P$ as the shortest path from $s$ to $v$. Let $x$ be the last vertex on $P$ such that $x\in \Bundle(u_j)$ for some $j < t$, and $y$ be the next one after $x$ on $P$, hence $y \in \Bundle(u)$ for some $u \in R \backslash \set{u_\ell}_{\ell < t}$. By Property \ref{prop:dv-correct} of induction hypothesis, $d(x) = \dis(s, x)$ after Step 1 of $j$-th iteration. Same as above we can show that $d(y) = \dis(s, x) + \dis(x, y) = \dis(s, y)$ and $d(u) \leq \dis(s, y) + \dis(y, u)$ (where $u=\bundle(y)$) before $t$-th iteration.
        We have:
        \begin{align*}
            \dis(s, y) \geq ~ & d(u) - \dis(y, u).\\
            \dis(s, v) < ~ & d(v) \tag{Assumption}\\
            \leq ~ & d(u_t) + \dis(u_t, v) \tag{$d(v)$ updated in Line \ref{line:bundle-Dijkstra:update-uv}}.
        \end{align*}
        
        On the other hand, $d(u) \geq d(u_t)$ after $t$-th iteration by Property \ref{prop:heap-pop-increase}. Since $d(u_t)$ doesn't change (Property \ref{prop:du-correct-before} and Observation \ref{obs:d-upper-bound}), and $d(u)$ can only decrease in $t$-th iteration, so $d(u) \geq d(u_t)$ holds throughout $t$-th iteration. Hence:
        \begin{equation}
        \label{eq:hence-equation}
            \dis(y, v) = \dis(s, v) - \dis(s, y) < \dis(u_t, v) + \dis(y, u),
        \end{equation}
        while the equation holds since $y$ lies on the shortest path from $s$ to $v$.

        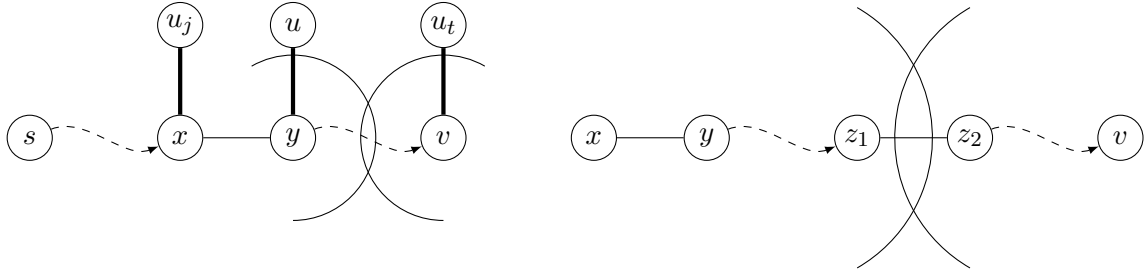
\begin{figure}[ht]
            \centering
            \begin{tikzpicture}
                \tikzstyle{mycircle}=[circle, minimum size = 6mm]
                \node [mycircle, draw, label=center:$s$] (s) at (-1.5, 0) {};
                \node [mycircle, draw, label=center:$x$] (x) at (0.5, 0) {};
                \node [mycircle, draw, label=center:$y$] (y) at (2, 0) {};
                \node [mycircle, draw, label=center:$u_j$] (uj) at (0.5, 1.5) {};
                \node [mycircle, draw, label=center:$u$] (u) at (2, 1.5) {};
                \node [mycircle, draw, label=center:$v$] (v) at (4, 0) {};
                \node [mycircle, draw, label=center:$u_t$] (ut) at (4, 1.5) {};
                \draw [-latex, dashed] (s) .. controls (-0.75, 0.3) and (-0.25, -0.3) .. (x);
                \draw (x) -- (y);
                \draw [ultra thick] (x) -- (uj);
                \draw [ultra thick] (y) -- (u);
                \draw [ultra thick] (v) -- (ut);
                \draw [-latex, dashed] (y) .. controls (2.75, 0.3) and (3.25, -0.3) .. (v);
                \draw (2, -1.1) arc(-90: 120: 1.1);
                \draw (4, -1.1) arc(270: 60: 1.1);
                
                \node [mycircle, draw, label=center:$x$] (x2) at (6, 0) {};
                \node [mycircle, draw, label=center:$y$] (y2) at (7.5, 0) {};
                \node [mycircle, draw, label=center:$z_1$] (z1) at (9.5, 0) {};
                \node [mycircle, draw, label=center:$z_2$] (z2) at (11, 0) {};
                \node [mycircle, draw, label=center:$v$] (v2) at (13, 0) {};
                
                \draw (x2) -- (y2);
                \draw [-latex, dashed] (y2) .. controls (8.25, 0.3) and (8.75, -0.3) .. (z1);
                \draw (z1) -- (z2);
                \draw [-latex, dashed] (z2) .. controls (11.75, 0.3) and (12.25, -0.3) .. (v2);
                \draw (9.5, 1.732) arc(60: -60: 2);
                \draw (11, 1.732) arc(120: 240: 2);
            \end{tikzpicture}
            \caption{Left: Illustration of the shortest path from $s$ to $v$ if $d(v) > \dis(s,v)$. Right: A closer look at path $x$ to $v$ if $\dis(y, v) \geq \dis(u_t, v)$.}
        \end{figure}
 
        Therefore there are two possible cases:
        \begin{itemize}
            \item $\dis(y, v) < \dis(u_t, v)$.
            
            In this case $y \in \Ball(v)$, so we can update $d(v)$ to $\dis(s,y)+\dis(y,v)$ on Line \ref{line:bundle-Dijkstra-update-yv}, contradicting to $d(v)>\dis(s,v)$.
            
            \item $\dis(y, v) \geq \dis(u_t, v)$.
            
            First, by Inequality \eqref{eq:hence-equation}, $\dis(y, u) > \dis(y, v) - \dis(u_t, v) \geq 0$. Let $z_1$ be the last vertex on path $P$ with $\dis(y, z_1) < \dis(y, u)$, we have $z_1 \in \Ball(y)$.
            Let $z_2$ be the next vertex on the path, then $\dis(y, z_2)\geq \dis(y, u)$, so $\dis(z_2,v) = \dis(y, v) - \dis(y, z_2) \leq \dis(y, v) - \dis(y, u) < \dis(u_t,v)$, that is, $z_2\in \Ball(v)$.
            (If $z_2$ does not exist, then $z_1=v$.)
            
            By Property \ref{prop:dv-correct} of induction hypothesis, $d(x) = \dis(s, x)$ just after Step 1 of $j$-th iteration, so $d(z_1)$ is updated to $\dis(s, z_1)$ in Line \ref{line:bundle-Dijkstra:update-xz1} of $j$-th iteration. Therefore $d(v)$ is updated to $\dis(s, v)$ in Line \ref{line:bundle-Dijkstra:update-z1v} of $t$-th iteration (since $j < t$), contradicting the assumption.
        \end{itemize}

        Therefore $d(v)=\dis(s,v)$ for all $v \in \Bundle(u_t)$ after Step 1 of $t$-th iteration.
    \end{enumerate}

\end{proof}

\begin{remark}\label{rmk:lower-bound}

Pettie and Ramachandran~\cite{PR05} proved that, any hierarchy-based SSSP algorithm on undirected graphs in comparison-addition model takes time at least $\Omega(m + \min\{n\log\log r, n\log n\})$, where $r$ is the ratio of the maximum-to-minimum edge weight. This bound becomes $\Omega(m+n\log n)$ when $r$ is exponentially large. Here the hierarchy-based algorithm is defined to generate a permutation $\pi_s$ satisfying \emph{hierarchy property:} $\dis(s, v)\geq \dis(s, u) + \sep(u, v) \Rightarrow \pi_{s}(u) < \pi_{s}(v)$, where $\sep(u, v)$ is the longest edge on the MST path between $u$ and $v$. The permutation $\pi_s$ is, though not defined to be, typically for the algorithms discussed in \cite{PR05}, the order that the algorithm visits the vertices. However, that is a worst-case lower bound, and our algorithm is randomized. Also the order that our algorithm visits the vertices \emph{does not} follow the hierarchy property: think of two vertices $x$ and $y$ are both connected to $u$ by edges $(x,u), (y,u)$ both with weight 1, and $x,y$ are both bundled to $u$. It is possible that $\sep(x,y)=1$ and $\dis(s, y)= \dis(s, x) + 2$ but we set no limit on the order we visit $x$ and $y$, that it is possible we visit $y$ before $x$. This explains why our algorithm can break this $\Omega(m + n\log n)$ lower bound in \cite{PR05}.

\end{remark}

\section{Improved Bundle Construction}
\label{sec:improved-construction-of-bundles}

In this section we propose an improved bundle construction that runs in $O(m\sqrt{\log n\cdot \log\log n})$ time with high probability.
In Section~\ref{sec:correctness-analysis} we showed that correctness of Bundle Dijkstra does not depend on the choice of $R$, as long as $\bundle(\cdot), \Ball(\cdot)$ and $\Bundle(\cdot)$ are correctly computed with respect to $R$. The running time for the bundle construction is $O(\sum_{v\in V\setminus R} \abs{S_v}\log \abs{S_v})$, and Bundle Dijkstra is $O(\sum_{v\in V\setminus R}\abs{\Ball(v)} + \abs{R}\log n)$.

Naturally, $\abs{S_v}$ is a random variable following geometric distribution for each vertex $v\in V$, and they are not independent since a vertex $x\in V$ may appear in several sets. However, for a subset $W \subseteq V$, if any vertex appears at most once in $\set{S_v}_{v\in W}$, the corresponding random variables $\set{\abs{S_v}}_{v\in W}$ are independent. By Lemma~\ref{lem:bounded-covariance-sum-distribution}, if each random variable is dependent to few other variables, their summation deviates from the expectation with exponentially small probability. So we manually include all those vertices with $\abs{S_v} \geq k\log k$ into $R$. In this way for each vertex in $V\setminus R$, its random variable is dependent to only a limited number of other ones, and we can bound their summation with high probability.

We introduce how to generate $R$ and compute $\{\bundle(v)\}_{v\in V\setminus R}$ below, as well as $\{\Ball(v)\}_{v\in V\setminus R}$, $\{\Bundle(u)\}_{u\in R}$ and $\dis(v, u)$ for $u\in \Ball(v)\cup\{\bundle(v)\}$. The pseudocode is given in Algorithm~\ref{alg:bundle-construction}. We still set parameter $k = \sqrt{\frac{\log n}{\log\log n}}$ as in Section~\ref{sec:algorithm}.

\paragraph{Improved Bundle Construction.}
\begin{itemize}
    \item Sample each vertex $v\in V\setminus \{s\}$ with probability $\frac{1}{k}$ to form set $R_1$ and add $s$ into $R_1$;
    \item For each $v\in V\setminus R_1$, run Dijkstra algorithm started from $v$, until we have extracted a vertex in $R_1$; or have already popped $k\log k$ vertices.
    \item In the former case, denote the extracted vertices in the order they appeared as list $V_{extract}^{(v)}$. Note that $V_{extract}^{(v)}$ is similar to $S_v$ of Section~\ref{sec:algorithm}.
    In the latter case, add $v$ into $R_2$;
    \item Set $R = R_1\cup R_2$, and for $v\in V\setminus R$, let the first vertex in $V_{extract}^{(v)}$ that lies in $R$ be $\bundle(v)$;
    \item With the results above, compute $\Bundle(u)$ for $u\in R$, $\Ball(v)$ for $v\in V\setminus R$, and record $\dis(v, u)$ for $u\in \Ball(v)\cup\{ \bundle(v)\}$. This step takes linear time.
\end{itemize}

The correctness of this bundle construction follows from the Dijkstra's algorithm~\cite{Dij59}. We only need to analyze $\abs{R}$, $\sum_{v\in V\setminus R}\abs{\Ball(v)}$ and its running time. The performance of this improved bundle construction is characterized in Lemma~\ref{lem:correctness-alg-truncated_dijkstra} below. By Lemma~\ref{lem:correctness-alg-truncated_dijkstra}, the bundle construction takes $O(mk\log k)$ time, and bundle Dijkstra takes $O(\sum_{v\in V\setminus R}\abs{\Ball(v)} + \abs{R}\log n) = O(mk+m\log n/k)$ with probability $1 - e^{-\Omega(n^{1 - o(1)})}$. Thus the total running time of our algorithm is $O(mk\log k + m\log n/k) = O(m\sqrt{\log n\cdot \log\log n})$ w.h.p. The proof of Lemma~\ref{lem:correctness-alg-truncated_dijkstra} is based on Lemma~\ref{lem:bounded-covariance-sum-distribution}.

\begin{algorithm}[ht]
    \caption{\textsc{BundleConstruction}$(G, s, k)$}
    \label{alg:bundle-construction}
	\SetKwInOut{Input}{Input}\SetKwInOut{Output}{Output}
	\Input{A graph $G=(V,E,w)$, the source vertex $s\in V$ and parameter $k$}
    \Output{All $R$, $\bundle(\cdot)$, $\Ball(\cdot)$, $\Bundle(\cdot)$ as described above}

    Initialize $R_1\gets\{s\}$, and insert each $v\in V\backslash\{s\}$ to $R_1$ independently with probability $\frac{1}{k}$\;
    Initialize $R_2 \gets \emptyset$\;
    \ForEach(\tcp*[f]{Truncated Dijkstra}){$v\in V\backslash R_1$}{
        Initialize Fibonacci heap $H^{(v)}$ with vertex $v$ and its key     $d^{(v)}(v) \gets 0$\;
        Initialize an ordered list $V^{(v)}_{extract}\gets ()$\;
        \While{$H^{(v)}$ is not empty}{
            $u\gets H.\textsc{ExtractMin}()$\;
            $V^{(v)}_{extract}.\textsc{Append}(u)$\;
            \lIf{$u\in R_1$}{quit the \textbf{while-loop}}
            \lElseIf{$\abs{V^{(v)}_{extract}} > k\log k$}{set $R_2 \gets R_2\cup \{v\}$ and quit the \textbf{while-loop}}
            \Else{
                \ForEach{$x\in N(u)$}{
                    \lIf{$x\notin H^{(v)}$ and $x\notin V^{(v)}_{extract}$}{$H^{(v)}.\textsc{Insert}(x, d^{(v)}(u)+w_{ux})$}
                    \lElseIf{$d^{(v)}(x) > d^{(v)}(u) + w_{ux}$}{$H.\textsc{DecreaseKey}(x, d^{(v)}(u) + w_{ux})$}
                }
            }
        }
    }
    $R \gets R_1\cup R_2$\;
    \ForEach{$v\in V\setminus R$}{
    $\bundle(v)\gets $ the first vertex in the ordered list $V^{(v)}_{extract}$ that lies in $R$\;
    }
    Compute $\{\Bundle(u)\}_{u\in R}$, $\{\Ball(v)\}_{v\in V\setminus R}$ and record $\dis(v, u)$ for $u\in \Ball(v)\cup \{\bundle(v)\}$\;
\end{algorithm}

\begin{lemma}
    \label{lem:correctness-alg-truncated_dijkstra}
    By running Algorithm~\ref{alg:bundle-construction}, with probability $1-e^{-\Omega(n^{1 - o(1)})}$, the following properties hold:
    \begin{enumerate}
        \item [(a)] $\abs{R} = O(\frac{m}{k})$.
        \item [(b)] $\sum_{v\in V\setminus R}\abs{\Ball(v)} = O(mk)$.
        \item [(c)] The running time of Algorithm~\ref{alg:bundle-construction} is $O(mk\log k)$.
    \end{enumerate}
\end{lemma}
\begin{proof}
    First, each vertex of $V\setminus \{s\}$ is inserted to $R_1$ independently with probability $\frac{1}{k}$, so by Chernoff bound, with probability $1 - O(e^{-m/k}) = 1 - e^{-\Omega(n^{1 - o(1)})}$, $\abs{R_1} = \Theta(m/k)$, and meanwhile $m' := \abs{V\setminus R_1} = \Theta(m)$.

    For each vertex $v\in V\setminus R_1$, define $X_v = \mathbb{I}[v\in R_2]$ and $Y_v = \abs{V_{extract}^{(v)}}$. Then, $X_v$ is a Bernoulli random variable, $X_v\in [0, 1]$ with probability $1$ and $\mathbb{E}[X_v] = (1 - \frac{1}{k})^{k\log k} = \Theta(\frac{1}{k})$. And $Y_v$ is a geometric random variable except its value is zero when $X_v=1$, so $Y_v\in [0, k\log k]$ with probability $1$ and $\mathbb{E}[Y_v] = k - (k+k\log k)(1 - \frac{1}{k})^{k\log k} = \Theta(k)$.\footnote{Detailed calculation: $\mathbb{E}[Y_v] = k - \sum_{i=k\log k+1}^{+\infty}\frac{1}{k}(1-\frac{1}{k})^{i-1}\cdot i = k - (1-\frac{1}{k})^{k\log k}\sum_{i=1}^{+\infty}\frac{1}{k}(1-\frac{1}{k})^{i-1}\cdot(i+k\log k)=k-(1-\frac{1}{k})^{k\log k}(k + k\log k)$. Noticing that $(1-\frac{1}{k})^{k\log k}\leq 1/k$, so $k - 1 - \log k \leq \mathbb{E}[Y_v] \leq k$ and $\mathbb{E}[Y_v] = \Theta(k)$.}

    For each vertex $v\in V\setminus R_1$, denote $V_{full}^{(v)}$ as the first $k\log k$ vertices extracted in the Dijkstra algorithm if it did not truncate. They are determined by the structure of $G$, so there is no randomness in $V_{full}^{(v)}$. The values of $X_v$ and $Y_v$ are determined by whether vertices in $V_{full}^{(v)}$ were inserted into $R_1$. Therefore, if $V_{full}^{(v_1)}, V_{full}^{(v_2)}, \cdots, V_{full}^{(v_j)}$ are disjoint, then $X_{v_1}, X_{v_2}, \cdots, X_{v_j}$ are independent, and similarly, $Y_{v_1}, Y_{v_2}, \cdots, Y_{v_j}$ are independent.
    
    For each vertex $w\in V_{full}^{(v)}$, because $w$ is found by $v$ within $k\log k$ steps of Dijkstra's algorithm, there must exist a path from $v$ to $w$ of no more than $k\log k$ edges, so by constant degree property, there are at most $3\cdot(1 + 2 + \cdots + 2^{k\log k - 1})\leq 3\cdot 2^{k\log k}$ different $u$ such that $w\in V_{full}^{(u)}$.
    Hence, for each $v$, there are at most $3k\log k \cdot 2^{k\log k}= O(n^{o(1)})$ different $u\in V\setminus R_1$ such that $V_{full}^{(v)}\cap V_{full}^{(u)} \neq \emptyset$.

    To apply Lemma~\ref{lem:bounded-covariance-sum-distribution}, for each $v\in R_1$, also define $X_v$, $Y_v$ and $V_{full}^{(v)}$ in the same way as if $v$ is executed in the loop of Algorithm~\ref{alg:bundle-construction}.

    Now, we apply Lemma~\ref{lem:bounded-covariance-sum-distribution} for $\{X_v\}_{v\in V}$ and $\{Y_v\}_{v\in V}$. For $\{X_v\}_{v\in V}$, $S$ is $V$ and $\abs{V} = m$, $\mu = \Theta(\frac{1}{k})$, $b = 1$, $T = O(n^{o(1)})$, and $\{W_v\}_{v\in V}$ are $\{ V_{full}^{(v)}\}_{v\in V}$, and we can verify that $8Tb\mu^{-1} = O(n^{o(1)})$ and $8b^3T/\mu^3 = O(n^{o(1)})$, so with probability at least $1 - e^{-\Omega(m/n^{o(1)})}$, it holds that $\sum_{v\in S}X_v = \Theta(m/k)$. And for $\{Y_v\}_{v\in V}$, $S$ is $V$, $\mu = \Theta(k)$, $b = k\log k$, $T = O(n^{o(1)})$, and $\{W_v\}_{v\in V}$ are also $\{ V_{full}^{(v)}\}_{v\in V}$, and similarly we infer that with probability $1 - e^{-\Omega(m/n^{o(1)})}$, it holds that $\sum_{v\in S}Y_v = \Theta(mk)$. Thus, we conclude that with probability $1 - e^{-\Omega(n^{1 - o(1)})}$, $\sum_{v\in V\setminus R_1}X_v = \Theta(m/k)$ and $\sum_{v\in V\setminus R_1} Y_v = \Theta(mk)$.

    Then, we prove the three claims of this lemma.

    For $(a)$, by definition $\abs{R} = \abs{R_1} + \abs{R_2}$, so by union bound, with probability $1 - e^{-\Omega(n^{1 - o(1)})}$, $\abs{R} = \abs{R_1}+ \sum_{v\in V\setminus R_1}X_v=O(\frac{m}{k})$.

    For $(b)$, by definition $\abs{\Ball(v)} \leq Y_v$, so $\sum_{v\in V\setminus R}\abs{\Ball(v)}\leq \sum_{v\in V\setminus R_1}Y_v$. Thus, with probability $1 - e^{-\Omega(n^{1 - o(1)})}$, $\sum_{v\in V\setminus R}\abs{\Ball(v)} = O(mk)$.

    For $(c)$, we count the total time for the truncated Dijkstra algorithm in all iterations. For each vertex $v\in V\setminus R_1$, by constant degree property, the number of $\textsc{Insert}$ operations is $O(Y_v)$, so $\abs{H^{(v)}} = O(Y_v) = O(k\log k)$. Therefore, each $\textsc{ExtractMin}$ operation takes time $O(\log(Y_v)) = O(\log k)$, and every other operation takes constant time. Thus the truncated Dijkstra algorithm of $v$ takes time $O(Y_v\log k)$. Thus, with probability $1 - e^{-\Omega(n^{1 - o(1)})}$, the total time of Algorithm~\ref{alg:bundle-construction} is $O(\sum_{v\in V\setminus R_1}Y_v\log k)=O(mk\log k)$.
\end{proof}

\begin{lemma} (Similar arguments as in~\cite{Janson})
    \label{lem:bounded-covariance-sum-distribution}
    Suppose a set of random variables $\{Z_v\}_{v\in S}$ satisfy that for each $v\in S$, $\mathbb{E}[Z_v] = \mu$, $Z_v\in[0, b]$ with probability 1, and each $Z_v$ is corresponded to a fixed deterministic set $W_v$ such that, if $W_{v_1}, W_{v_2}, \cdots, W_{v_j}$ are disjoint, then $Z_{v_1}, Z_{v_2}, \cdots, Z_{v_j}$ are independent, and $W_v$ intersects with at most $T$ different other $W_u$.

    Then, with probability at least $1 - 8Tb\mu^{-1}\cdot e^{-\frac{\mu^3\abs{S}}{8b^3T}}$, it holds that $\sum_{v\in S}Z_v = \Theta(\abs{S}\mu)$.
\end{lemma}
\begin{proof}
We try to partition $\{Z_v\}_{v\in S}$ into several subsets $\{\mathcal{Z}_t\}$ such that all $Z_v$'s in each $\mathcal{Z}_t$ are independent so that we can apply Hoeffding's inequality, or the size of $\mathcal{Z}_t$ is small so that we can bound them by the upper bound $b$, and finally combine everything by the union bound. Also note that we do not need to actually compute $\{\mathcal{Z}_t\}$, as they are merely introduced for this mathematical proof.

Fix parameter $p = \frac{\abs{S}\mu}{4Tb}$. Since each $W_v$ intersects with at most $T$ different other $W_u$, whenever there are at least $p(T+1)$ elements in $\{Z_v\}_{v\in S}$, we can pick $p$ different $Z_v$ from them whose $W_v$ are disjoint, so that they are independent: pick an arbitrary $Z_v$ and discard those $Z_u$ if $W_u\cap W_v\neq \emptyset$; since there are at most $T$ such $Z_u$ different from $Z_v$, every time we discard at most $T+1$ elements. So from $p(T+1)$ elements we can pick $p$ of them.
We let them form a $\mathcal{Z}_t$ and remove them from $\{Z_v\}_{v\in S}$. Repeating this process we will end up with a partition $\{\mathcal{Z}_1, \mathcal{Z}_2, \cdots, \mathcal{Z}_q, \mathcal{Z}_{q+1}\}$ of $\{Z_v\}_{v\in S}$ such that: $\abs{\mathcal{Z}_t} = p$, and all $Z_v\in \mathcal{Z}_t$ are independent for $1\leq t\leq q$; $\abs{\mathcal{Z}_{q+1}} \leq p(T+1)\leq 2pT = \frac{\mu}{2b}\abs{S}$. By definition $\mu \leq b$, so $\abs{\mathcal{Z}_{q+1}}\leq \frac{1}{2}\abs{S}$.

Then by Hoeffding's inequality, for each $1\leq t\leq q$, 
\[\Pr\left[\abs{\sum_{v\in \mathcal{Z}_t}Z_v - \abs{\mathcal{Z}_t}\mu} \geq \frac{1}{2}\abs{\mathcal{Z}_t}\mu\right] \leq 2e^{-\frac{2(\frac{1}{2}\abs{\mathcal{Z}_t}\mu)^2}{\abs{Z_t}b^2}} = 2e^{-\frac{\mu^2p}{2b^2}}. \]
and $0\leq \sum_{v\in \mathcal{Z}_{q+1}}Z_ v\leq \abs{\mathcal{Z}_{q+1}}b$ with probability $1$.

By union bound, with probability at least $1 - 2qe^{-\frac{\mu^2p}{2b^2}}$,
\[\sum_{v\in S}Z_v \geq \frac{1}{2}\sum_{t=1}^q\abs{\mathcal{Z}_t}\mu = \frac{1}{2}(\abs{S} - \abs{\mathcal{Z}_{q+1}})\mu \geq \frac{1}{2}\left(\abs{S}-\frac{1}{2}\abs{S}\right)\mu = \frac{1}{4}\abs{S}\mu,\]
and meanwhile
\[\sum_{v\in S}Z_v\leq \frac{3}{2}\sum_{t=1}^q\abs{\mathcal{Z}_t}\mu + \abs{\mathcal{Z}_{q+1}}b\leq \frac{3}{2}\abs{S}\mu + \frac{\mu}{2b}\abs{S}\cdot b = 2\abs{S}\mu.\]

And from $\abs{S} \geq \sum_{t=1}^q\abs{\mathcal{Z}_t} \geq qp$, we conclude that $q \leq \abs{S}/p = 4Tb/\mu$. Thus, with probability at least $1 - 8Tb\mu^{-1}e^{-\frac{\mu^3\abs{S}}{8b^3T}}$, it holds that $\sum_{v\in S}Z_v = \Theta(\abs{S}\mu)$.
\end{proof}

\section{Discussion}\label{sec:disussion}
We gratefully acknowledge an anonymous reviewer for suggesting that constant-degree is not a necessary condition for this algorithm, so we can get improved time complexity when $m=\omega(n)$ and $m=o(n\log n)$. Instead of making the graph of degree 3, we use similar methods to split the vertices of degree $>m/n$ to vertices of degrees $\leq m/n$, so that the number of vertices is still $O(n)$. Then in each step:
\begin{itemize}
    \item In bundle construction, the time for Dijkstra search for every vertex $v$ will become $O(|S_v|\cdot \frac{m}{n}+|S_v|\log (|S_v|\cdot\frac{m}{n}))$, since the size of the heap is at most $|S_v|\cdot\frac{m}{n}$, so in total $O(mk+nk\log(mk/n))$.
    \item The time for Bundle Dijkstra will become $O(\frac{n}{k}\log n+ mk)$, since the number of vertices $z_1$ in Step 1 for every $v$ is $O(\frac{m}{n}|\Ball(v)|)$, and the number of vertices $z_1$ in Step 2 for every $y$ is $O(|\Ball(y)|)$ but each vertex appears $O(m/n)$ times as $y$ in Step 2.
    \item When $m/n=o(\log n)$, one can check that the analysis of independence in Section~\ref{sec:improved-construction-of-bundles} still works, since the number of different $u\in V\setminus R_1$ which have $V_{full}^{(v)}\cap V_{full}^{(u)} \neq \emptyset$ for each $v\in V\setminus R_1$ is still $O(n^{o(1)})$. 
\end{itemize}

Thus, the time complexity for this algorithm is $O(\frac{n}{k}\log n+ mk + nk\log(mk/n))$. When $m<n\log\log n$, $k$ still equals to $\sqrt{\frac{\log n}{\log \log n}}$, and the time bound is $O(n\sqrt{\log n\log\log n})$. When $n\log\log n\leq m< n\log n$, let $k=\sqrt{\frac{n}{m}\log n}$, and the time bound will be $O(\sqrt{mn\log n})$.

\bibliographystyle{plainnat}
\bibliography{reference}


\end{document}